\documentclass[leqno]{llncs}
\usepackage{amsmath}
\usepackage{amsfonts}
\usepackage{graphicx}
\usepackage{makeidx}  

\begin{document}
\title{A note on anti-coordination and social interactions\thanks{This research is  supported by the 973 Program (2010CB731405) and National Natural Science Foundation of China (71101140).}}
\author{Zhigang Cao,Xiaoguang Yang}
\institute{Key Laboratory of Management, Decision \& Information Systems,\\ Academy of Mathematics
and Systems Science,\\ Chinese Academy of Sciences, Beijing, 100190, China.\\
\email{zhigangcao@amss.ac.cn,xgyang@iss.ac.cn}}

\maketitle
\begin{abstract}
This note confirms a conjecture of [Bramoull\'{e}, Anti-coordination and social interactions, Games and Economic Behavior, 58, 2007: 30-49]. The problem, which we name the maximum independent cut problem, is a restricted version of the MAX-CUT problem, requiring one side of the cut to be an independent set. We show that the maximum independent cut problem does not admit any polynomial time algorithm with approximation ratio better than $n^{1-\epsilon}$, where $n$ is the number of nodes, and $\epsilon$ arbitrarily small, unless P=NP. For the rather special case where each node has a degree of at most four, the problem is still APX-hard.
\end{abstract}

\section{Introduction}
Bramoull\'{e} (2007) introduces the anti-coordination game on networks. This is a very promising model that captures the social interactions where two choices are strategic substitutes.  The game is described as follows. There is an undirected graph $G=(N,E)$, where each node stands for a player. Each pair of players that are neighbor to each other on this graph play the following anti-coordination game:

\begin{center}\begin{tabular}{|c|c|c|}
\hline
&A&B\\
\hline
A&$(\pi_{AA},\pi_{AA})$&$(\pi_{AB},\pi_{BA})$\\
\hline
B&$(\pi_{BA},\pi_{AB})$&$(\pi_{BB},\pi_{BB})$\\
\hline\end{tabular},\end{center}
 where $\pi_{BA}>\pi_{AA}$ and $\pi_{AB}>\pi_{BB}$.

 Obviously, this game is symmetric, and has two pure Nash equilibria, $(A,B)$ and $(B,A)$ (notice that it also has a unique mixed Nash equilibrium). The name ``anti-coordination" comes from the fact that each player likes to choose a different action with her opponent.

 In Bramoull\'{e}'s model, each player plays the anti-coordination game once with each of her neighbors, and she has to choose the same action among all the games she is involved. Her payoff in the whole game is simply defined as the sum of her payoffs in all the partial games involved.

 The relative payoffs, i.e. the gap between $\pi_{BA}$ and $\pi_{AA}$, and that between $\pi_{AB}$ and $\pi_{BB}$, matter a lot. Define $\pi_A=\pi_{AB}-\pi_{BB}$ and $\pi_B=\pi_{BB}-\pi_{AB}$. Players choosing action $A$ are called $A$ players, and those choosing $B$ are called $B$ players. Intuitively, $\pi_A$ is the payoff of the $A$ player at equilibrium $(A,B)$ minus the payoff obtained when she deviates, and $\pi_B$ is the payoff of the $B$ player at equilibrium $(A,B)$ minus the payoff obtained when she deviates.

  The following frustration function is a powerful tool to study the anti-coordination game on networks:\begin{equation*}\varphi(s,\pi_A,\pi_B,G)=\pi_{A}n_{BB}+\pi_{B}n_{AA},\end{equation*}
where $s$ is a pure action profile, $n_{AA}$ the number of edges between $A$ players, and $n_{BB}$ that between $B$ players. In fact, the opposite of the frustration function, $-\varphi$, is a potential function of the networked anti-coordination game. Therefore, the well-known result of Monderer and Shapley (1996) tells us that: (i) the networked anti-coordination game always has a pure Nash equilibrium, (ii) a pure action profile $s$ is a Nash equilibrium  if and only if it is a local maximum of the potential function, or equivalently a local minimum of the frustration function,  and (iii) the (asynchronous) best response dynamic leads any pure action profile to a Nash equilibrium. A more recent, and more general, result of Hoefer and  Suri (2009) shows that the best response dynamic terminates in $O(nm^2)$, where $n$ is the number of nodes and $m$ the number of edges. Consequently, finding an arbitrary pure Nash equilibrium for the networked anti-coordination game is polynomially solvable.

The main concentration of  Bramoull\'{e} (2007) is on the global minimums of the frustration function. They correspond to special Nash equilibria. Bramoull\'{e} observes that when $\pi_A=\pi_B$, minimizing the frustration function is generally NP-hard, because it is equivalent to the well-known MAX-CUT problem, which is NP-hard (cf. Garey and Johnson, 1979; or Korte and Vygen, 2008).

Bramoull\'{e} (2007) conjectures that in the polar case, where $\pi_A\gg \pi_B$, minimizing the frustration function is also NP-hard. He observes that in this case, a pure action profile is a Nash equilibrium if and only if all the $B$ players are independent on graph $G$ (because once a $B$ player has one $B$ neighbor, she will definitely deviate to $A$), and each $A$ player must have at least one  $B$ neighbor. The second condition requires that the set of $B$ players is a maximal independent set, because otherwise there must exist an $A$ player who has no $B$ neighbor (we assume w.l.o.g. that the graph $G$ is connected), who can get strictly better off by deviating. We name this problem {\it the maximum independent cut problem}, which,  in the language of combinatorial optimization,  is stated formally as follows.

\line(1,0){370}
\begin{center}The maximum independent cut problem\end{center}

{\bf Input.} A graph $G=(N,E)$.

{\bf Output.} A maximal independent set $C\subseteq N$.

{\bf Objective.} Maximizing the number of edges between $C$ and $N\setminus C$.

 \line(1,0){370}\\

 The maximum independent cut problem is obviously a restricted version of the MAX-CUT problem. Not surprisingly, as conjectured by Bramoull\'{e} (2007), it is NP-hard. In fact, we shall show in the next section that it is even hard to approximate, and does not admit a PTAS even for a rather special case where each node has at most four neighbors.
\section{The results}
\begin{theorem}The maximum independent cut problem is NP-hard.\end{theorem}

\begin{proof}We prove by reduction from the maximum independent set problem. Given an instance of the maximum independent set problem $G=(N,E)$, suppose $N=\{1,2,\cdots,n\}$. We assume w.l.o.g. that $G$ is not a complete graph. Construct an instance of the maximum independent cut problem $G'=(N',E')$ as follows: (a) $N'=N\cup \{n+1,n+2,\cdots, n^2+n\}$; (b) For any $1\leq i<j\leq n$, $i$ and $j$ are connected in $G'$ if and only if they are connected in $G$; (c) For all $n+1\leq i<j\leq n^2+n$, they are connected in $G'$; (d) For all $1\leq i\leq n$ and $n+1\leq j\leq n^2+n$, they are connected in $G'$. Let $C^*$ be a maximum independent cut of $G'$, we prove that it is also a maximum independent set of $G$. We claim that $C^*\subseteq N$. In fact, if there exists $n+1\leq i\leq n^2+n$ such that $i\in C^*$, then it must be true that $C^*=\{i\}$, because each node in $\{n+1,n+2,\cdots, n^2+n\}$ is connected to all other nodes in $G'$. And hence the objective value of this solution is $n^2+n-1$. However, the maximum independent set of $G$ has at least two members, because it is not a complete graph as we assume. If this set is chosen as a solution to $G'$, the corresponding objective value is at least $2n^2$, which is bigger than $n^2+n-1$. Suppose $C_1\subseteq N$ and $C_2\subseteq N$ are two independent sets of $G$. Because the number of neighbors that each node in $N$ has is at least $n^2$, and at most $n^2+n-1$, we know that the objective value of $C_1$ is greater than that of $C_2$ if and only if $C_1$ has a bigger cardinality than $C_2$. Therefore, to be an optimal solution, $C^*$ must be a maximum independent set of $G$. Hence the theorem.\qed\end{proof}

Recent results (H{\aa }stad, 1999; Khot, 2001; Zuckerman, 2006) show that the maximum independent set problem does not admit a polynomial time algorithm with approximation ratio less than $O(n^{1-\epsilon})$, where $\epsilon$ is arbitrarily small, unless $P=NP$. Their results are stated for the maximum clique problem, which is equivalent to the maximum independent set problem. Using the same reduction as in the proof to Theorem 1, it can be trivially shown that this result also holds for the maximum independent cut problem.

\begin{corollary}The maximum independent cut problem does not admit a polynomial time algorithm with approximation ratio less than $O(n^{1-\epsilon})$, where $\epsilon$ is arbitrarily small, unless $P=NP$.\end{corollary}

It can be observed also that the graph we construct in the above proof is rather dense. When the graph is sparse, the maximum independent cut problem  is still hard. In fact, we can prove that a rather sparse version is APX-hard, which tells us that this special case does not admit any PTAS, unless P=NP. 

To ease the presentation, we define the {\it 4-sparse maximum independent cut problem} as the restricted version of the maximum independent cut problem where each node has a degree of at most four.

 \begin{theorem}The 4-sparse maximum independent cut problem is MAXSNP-hard, and thus APX-hard.\end{theorem}

For the reader who is not familiar with the concepts of  APX-hard or MAXSNP-hard, or their relation, we refer her to Korte and Vygen (2008), Papadimitriou and Yannakakis (1991), and Khanna et. al. (1998).

We shall prove the above theorem by reduction from a restricted version of the MAX-SAT problem, the 3-OCC-MAX-2SAT problem. In this problem, each clause has exactly two literals, and  each literal occurs in at most three clauses. The 3-OCC-MAX-2SAT is MAXSNP-hard (Berman and Karpinski, 1995). This problem is stated formally as follows. Note that for any Boolean variable $x$, we use $\overline{x}$ to denote its negation.

 \line(1,0){370}
\begin{center}The 3-OCC-MAX-2SAT problem\end{center}

{\bf Input.} $n$ Boolean variables: $x_1,x_2,\cdots,x_n$, and $m$ clauses, each having two literals: $x_{11}\vee x_{12}, x_{21}\vee x_{22},\cdots, x_{m1}\vee x_{m2}$, where $\forall 1\leq j\leq m$, $x_{j1},x_{j2}\in \{x_1,\overline{x_1}, x_2,\overline{x_2},\cdots,x_n,\overline{x_n}\}$. For each literal $x$, it occurs in at most three clauses.

{\bf Output.} A truth assignment of the $n$ variables.

{\bf Objective.} Maximizing the number of true clauses.

 \line(1,0){370}\\

 We assume w.l.o.g. in the above problem that for each clause $x_{j1}\vee x_{j2}$, $x_{j1}$ and $x_{j2}$ are not negation to each other. Because otherwise, we can simply delete these kinds of clauses, without changing the problem at all. We also assume that for each variable $x_i$, both $x_i$ and its negation $\overline{x_i}$ occurs in at least one clauses, because if $\overline{x_i}$ does not occur at all, we can safely let $x_i=1$, and similarly if $x_i$ does not occur, we can let $x_i=0$.

\begin{proof} Given an instance $I$ of the 3-OCC-MAX-2SAT problem, construct an instance $G$ of the 4-sparse maximum independent cut problem as follows. (a) For each variable $x_i$ and its negation $\overline{x_i}$, there are two nodes $X_i$ and $\overline{X_i}$ corresponding to them, respectively. We call these nodes {\it chief nodes}. (b) For each pair $X_i$ and $\overline{X_i}$, there is an edge between them. (c) For each clause $x_{j1}\vee x_{j2}$, let $X_{j1}, X_{j2}\in \{X_1,\overline{X_1},X_2,\overline{X_2},\cdots, X_n,\overline{X_n}\}$ be the corresponding chief nodes for $x_{j1}$ and $x_{j2}$, respectively. Three {\it accessory} nodes, $Y_{j1}$, $Y_{j2}$ and $Y_{j3}$, are associated, in a way as illustrated in Fig. 1. Therefore, noting that $m\leq 3n$, there are $2n+3m\leq 2n+9n=11n$ nodes, and $n+5m\leq n+15n=16n$ edges in total.

\begin{figure}[t]
\centering
 \includegraphics[width=6cm]{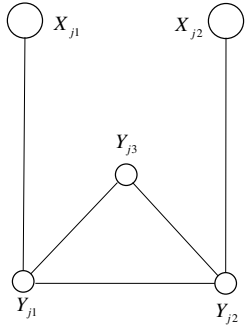}\\
  \caption{Gadget($X_{j1},X_{j2}$): chief nodes are denoted by large circles.}
\end{figure}

 First of all, notice that $G$ is indeed an instance of the 4-sparse maximum independent cut problem, because each node has a degree of at most four. We denote the above construction as $f$, i.e. $G=f(I)$.

Suppose $C$ is a solution to $G$. We construct $T$, a truth assignment of $I$ as follows: for each variable $x_j$, remember that $X_j$ is the corresponding chief node. If $X_j\in C$, then let $x_j=1$. Otherwise, let $x_j=0$. We call the above mapping $g$, i.e. $T=g(C)$.

We prove that $(f,g)$ is an L-reduction from the 3-OCC-MAX-2SAT problem to the 4-sparse maximum independent cut problem. We need to check three things:

 (a) $f$ and $g$ are both polynomially computable.

 (b) there exists a constant $\alpha$ such that $OPT(f(I))\leq \alpha OPT(I)$, where $OPT(I)$ is the optimal objective value of $I$, and $OPT(f(I))$ is the optimal objective value of $f(I)$.

 (c) there exists a constant $\beta$ such that $OPT(I)-u(I, g(C))\leq \beta (OPT(f(I))-v(f(I),C))$, where $u(I,g(C))$ is the objective value of solution $g(C)$ in instance $I$,  and $v(f(I),C)$ is the objective value of solution $C$ in instance $f(I)$.

Condition (a) can be checked trivially from our constructions. Condition (b) is also easy. In fact, we can safely choose $\alpha=32$, because (i) $OPT(I)\geq n/2$: for each variable $x_i$, if it occurs more than its negation $\overline{x_i}$, we let $x_i=1$, and otherwise we let $x_i=0$. Recall the assumption that both $x_i$ and $\overline{x_i}$ occur in at least one clauses. This truth assignment guarantees that there are at least $n$ true literals, and hence at least $n/2$ true clauses. The lower bound is obtained when the $n$ true literals are paired arbitrarily (suppose w.l.o.g. that $n$ is even). (ii) $OPT(f(I))\leq 16n$: this is an upper bound of the number of all edges in $f(I)$. So we are left to prove the last condition.

We analyze the structure of $C$, a solution to $G=f(I)$.  Recall by definition that $C$ is above all a maximal independent set. For each Gadget$(X_{j1},X_{j2})$, we discuss in four cases.

(i) If $X_{j1}\in C$ and $X_{j2}\in C$,  then it must be true that $Y_{j3}\in C$. Because $Y_{j1}$ is connected to $X_{j1}$, $Y_{j2}$ is connected to $X_{j2}$, neither of them can be selected by $C$. In order to be a maximal independent set, $C$ must include $Y_{j3}$. And in this case Gadget$(X_{j1},X_{j2})$ contributes 4 to the objective value.

 (ii) If $X_{j1}\in C$ and $X_{j2}\notin C$,  then it must be true that either $Y_{j2}\in C$ or $Y_{j3}\in C$, and Gadget$(X_{j1},X_{j2})$ contributes either 4 or 3 to the objective value.

 (iii) If $X_{j1}\notin C$ and $X_{j2}\in C$,  then it must be true that either $Y_{j1}\in C$ or $Y_{j3}\in C$, and Gadget$(X_{j1},X_{j2})$ contributes either 4 or 3 to the objective value.

 (iv) If $X_{j1}\notin C$ and $X_{j2}\notin C$,  then it must be true that either $Y_{j1}\in C$ or $Y_{j2}\in C$ or $Y_{j3}\in C$, and this gadget contributes either 3 (in the first two cases) or 2 (in the third case) to the objective value.

 According to the above discussion, we can observe that edges within gadgets contribute at most $4u(I,g(C))+3(m-u(I,g(C)))=3m+u(I,g(C))$. Since there are $n$ edges between chief nodes, we know that \begin{equation}v(f(I),C)\leq n+3m+u(I,g(C)).\label{1}\end{equation}

 Suppose now $C^*$ is an optimal solution to $f(I)$. Then, in each of  the first three cases discussed above, where at least one of the chief nodes is selected, Gadget$(X_{j1},X_{j2})$ contributes 4, and in the last case, where no chief node is selected,  Gadget$(X_{j1},X_{j2})$ contributes at most 3. Therefore, considering only edges within gadgets, $C^*$ is supposed to maximize the number of gadgets such that at least one of its chief nodes is selected.
We consider now the edges between chief nodes. First of all, for each $1\leq i\leq n$, it cannot happen that $X_i\notin C^*$ and  $\overline{X_i}\notin C^*$, because otherwise we can add $X_i$ (or $\overline{X_i}$) into $C^*$ and modify the selections in the affected gadgets to keep the new solution a maximal independent set. This will make things strictly better. Second of all, since $X_i$ and $\overline{X_i}$ are connected, they cannot be selected simultaneously. Therefore, exactly one node in $\{X_i,\overline{X_i}\}$ is selected by $C^*$, and hence it must be true that all the $n$ edges between chief nodes contribute to the objective value in $C^*$.
Based on the above discussions, we know that the number of gadgets that at least one of its chief nodes is selected by $C^*$ is maximized, and hence \begin{equation}OPT(f(I))=n+3m+OPT(I).\label{2}\end{equation}

Combining (\ref{1}) and (\ref{2}), we conclude that \begin{eqnarray*}&&OPT(f(I))-v(f(I),C)\\
&\geq& (n+3m+OPT(I))-(n+3m+u(I,g(C)))\\
&=&OPT(I)-u(I,g(C)),\end{eqnarray*}
and hence we can take $\beta=1$.\qed\end{proof}

By equation (\ref{2}), we know that $OPT(f(I))\leq n+9n+OPT(I)\leq 20OPT(I)+OPT(I)=21OPT(I)$. Therefore, we can take $\alpha=21$ to get a tighter lower bound for the inapproximability of the maximum independent cut problem. We note first that there is no polynomial algorithm for the 3-OCC-MAX-2SAT problem with approximation ratio less than $2012/2011\doteq 1+5\cdot 10^{-4}.$ This result is also proved by  Berman and Karpinski (1995).
\begin{corollary}The maximum independent cut problem with maximum degree of four does not admit a polynomial algorithm with approximation ratio less than $1+1/42231\doteq 1+2\cdot 10^{-5}$, unless P=NP.\end{corollary}
\begin{proof}Suppose not, and let $\mathcal{A}$ be an algorithm with approximation ratio less than $1+1/42231$. Let the algorithm of the 3-OCC-MAX-2SAT problem derived by the L-reduction $(f,g)$ in the proof to Theorem 2 be $\mathcal{A}(f,g)$. For any instance $I$, we have \begin{eqnarray*}&&OPT(I)-\mathcal{A}(f,g)(I)\\
&\leq&OPT(f(I))-\mathcal{A}(f(I))\\
&<&(1/42231) OPT(f(I))\\
&\leq&(1/42231)\cdot 21OPT(I)\\
&=&(1/2011) OPT(I).\end{eqnarray*} \end{proof}

This implies that $OPT(I)<(2012/2011)\mathcal{A}(f,g)(I)$, a contradiction with the result of  Berman and Karpinski (1995) .\qed
\section{Concluding remarks}
In this note, we prove a conjecture of Bramoull\'{e} (2007) that a new combinatorial optimization problem, the maximum independent cut problem, is NP-hard. This confirms the insight of Bramoull\'{e} that minimizing frustration is computationally hard even in the polar case. It can be observed that none of the constructed graphs, neither in the proof to Theorem 1 nor in that to Theorem 2, are claw-free. Since it is known that the maximum independent set problem, in fact also its weighted version, is polynomially solvable for claw-free graphs (Sbihi, 1980, Minty, 1980; see also the survey of Faudree, Flandrin and Ryj$\acute{a}\breve{c}$ek, 1999), it's very interesting for future research to consider the claw-free graphs for the maximum independent cut problem.

Other interesting directions include: (i) To study other cases of the networked anti-coordination game, where $\pi_A\neq \pi_B$, and neither $\pi_A\gg \pi_B$ nor $\pi_B\gg \pi_A$. (ii) To study the PoA and PoS for the networked anti-coordination game. (iii) To design approximate algorithms for various special cases of the maximum independent cut problem.

{\bf Acknowledgements.} The authors are grateful to two anonymous reviewers for pointing out two mistakes in the original proofs to Theorem 1 and Theorem 2, and their valuable suggestions on presentation of this paper.


\begin{thebibliography} {9}
\bibitem{b07}Y. Bramoull\'{e}. Anti-coordination and social interactions. Games and Economic Behavior, 58, 2007, 30-49.
\bibitem{bk95}P. Berman and M. Karpinski. On Some Tighter Inapproximability Results, Lecture Notes in Computer Science, 1644, 1999: 200-209.
\bibitem{ffr96}R. Faudree, E. Flandrin, Z. Ryj$\acute{a}\breve{c}$ek. Claw-free graphs-a survey, Discrete Mathematics, 164, 1997: 87-147.
\bibitem{gj79}M.R. Garey  and D.S. Johnson. Computers and Intractability: A Guide to the
Theory of NP-Completeness. 1979, Freeman, San Francisco.
\bibitem{h99}J. H{\aa}stad. Clique is hard to approximate within $n^{1-\epsilon}$, Acta Mathematica, 182(1), 1999: 105-142.
\bibitem{hs09}M. Hoefer1 and S. Suri. Dynamics in Network Interaction Games, Lecture Notes in Computer Science, 5805, 2009: 294-308.
\bibitem{k01}S. Khot, Improved inapproximability results for MaxClique, chromatic number and approximate graph coloring, Proc. 42nd IEEE Symp. Foundations of Computer Science, 2001: 600-609.
\bibitem{kmsv98}S. Khanna, R. Motwani, M. Sudan, and U. Vazirani. On Syntactic versus Computational Views of Approximability, SIAM Journal on Computing, 28, 1998: 164-191.
\bibitem{kv08}B. Korte and J. Vygen. Combinatorial Optimization: Theory and Algorithms, fourth edition, 2008, Springer-Verlag Berlin Heidelberg.
\bibitem{m80}G.J. Minty. On maximal independent sets of vertices in claw-free graphs. Journal of  Combinatorial Theory, Serial B, 28, 1980: 284-304.
\bibitem{ms96}D. Monderer, and L. Shapley.  Potential games. Games and Economic Behavior. 24, 1996: 124-143.
\bibitem{py91}C.H. Papadimitriou and M. Yannakakis. Optimization, approximation, and complexity
classes. Journal of Computer and System Sciences 43, 1991: 425-440
\bibitem{s80}N. Sbihi. Algorithme de recherche d'un stable de cardinalit¨¦ maximum dans un graphe sans ¨¦toile (in French), Discrete Mathematics 29 (1), 1980: 53-76.
\bibitem{z06}D. Zuckerman. Linear degree extractors and the inapproximability of max clique and chromatic number, Proc. 38th ACM Symp. Theory of Computing, 2006: 681-690.
\end{thebibliography}
\end{document}